\documentclass[reprint]{revtex4-1}

\usepackage[english]{babel}
\usepackage[utf8]{inputenc}
\usepackage{amsmath}
\usepackage{subfigure}
\usepackage{floatrow}
\usepackage{graphicx}
\usepackage{amsthm}
\usepackage{amssymb}

\makeatletter
\renewcommand{\p@subsection}{}
\renewcommand{\p@subsubsection}{}
\makeatother

\newtheorem{problem}{Problem}[section]
\newtheorem{proposition}{Proposition}[section]
\newtheorem{theorem}{Theorem}[section]
\newtheorem{corollary}{Corollary}[section]
\newtheorem{definition}{Definition}[section]

\begin{document}

\title{Interception in Distance-Vector Routing Networks}

\author{David Burstein}
\affiliation{University of Pittsburgh}
\author{Franklin Kenter}
\affiliation{Rice University}
\author{Jeremy Kun}
\affiliation{University of Illinois at Chicago}
\author{Feng Shi\footnote{Corresponding author: bill10@uchicago.edu}}
\affiliation{University of Chicago}

\begin{abstract}
Despite the large effort devoted to cybersecurity research over the last
decades, cyber intrusions and attacks are still increasing. With respect to
routing networks, route hijacking has highlighted the need to reexamine the
existing protocols that govern traffic routing. In particular, our primary
question is how the topology of a network affects the susceptibility of a
routing protocol to endogenous route misdirection. In this paper we define and
analyze an abstract model of traffic interception (i.e. eavesdropping) in
distance-vector routing networks. Specifically, we study algorithms that
measure the potential of groups of dishonest agents to divert traffic through
their infrastructure under the constraint that messages must reach their
intended destinations. We relate two variants of our model based on the allowed
kinds of lies, define strategies for colluding agents, and prove optimality in
special cases. In our main theorem we derive a provably optimal monitoring
strategy for subsets of agents in which no two are adjacent, and we extend this
strategy to the general case. Finally, we use our results to analyze the
susceptibility of real and synthetic networks to endogenous traffic
interception. In the Autonomous Systems (AS) graph of the United States, we
show that compromising only 18 random nodes in the AS graph surprisingly
captures 10\% of all traffic paths in the network in expectation when a
distance-vector routing protocol is in use.
\end{abstract}

\keywords{Traffic interception, Routing networks, Distance-vector routing protocols, Distance fraud}

\maketitle

\section{Introduction}
Several recent events have demonstrated that internet routing protocols are
particularly vulnerable to misdirections in routing \cite{GormanV13, JimCowie,
Madory, Goodin}. This brings up the question: How vulnerable are trust-based
communication protocols to malicious agents who can abuse this trust?

\begin{figure*}
\centering
\includegraphics[width=1 \textwidth]{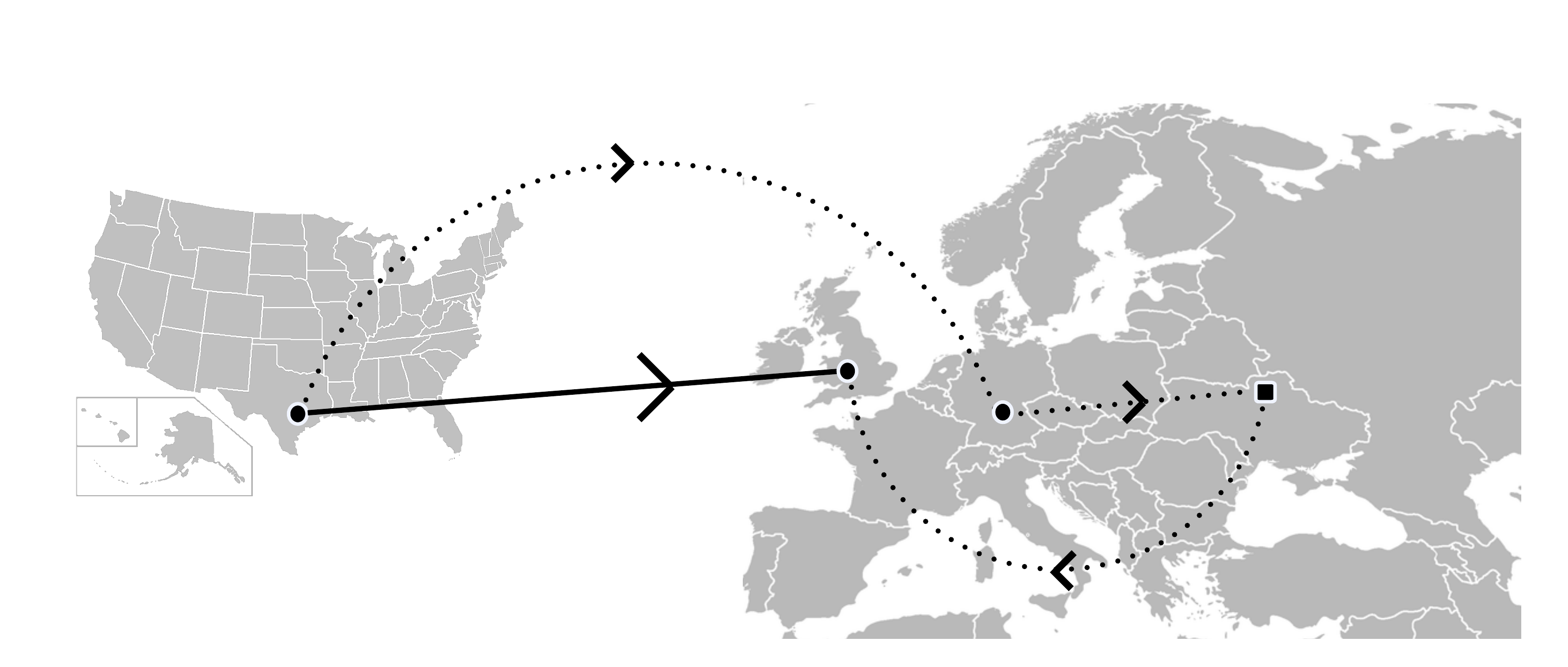}
\caption{An incident in connection with traffic interception as reported in
\cite{Madory}, where a ``colluding agent'' (square) in Ukraine, broadcasts a
false distance and intercepts a message sent from an honest agent in Texas
whose target is the British Telecom, that services the UK Atomic Weapon's
Establishment. The solid line indicates how the message should be routed under
the presence of only honest agents.  The dotted line demonstrates how the
message is routed when the ``colluding agent'' (square) broadcasts a false
distance, and reroutes the message to its intended target.  The above image is
a simplification of the documented rerouting, as the circle in the UK is in
fact representative of multiple honest agents (autonomous systems).}
\label{fig:rerout} 
\end{figure*}

To answer this question, we present and study a model of traffic interception
in routing networks largely inspired by real-life incidents. See Figure
\ref{fig:rerout} for an example. In practice, a wide range of routing protocols
are developed and implemented for communications. For example, the Border
Gateway Protocol (BGP) is used for inter-domain communications between
autonomous systems in the modern Internet; distance-vector and link-state
routing protocols are used for intra-domain communications. Instead of going
into the technical details of the protocols and their business models, we study
a mathematical and parsimonious model for traffic interception in general. The
model closely resembles the distance-vector protocols such as the Routing
Information Protocol (RIP) \cite{Hedrick, Malkin}. Due to their simplicity and
efficiency, distance-vector protocols are widely used in moderate-size IP
networks and ad hoc mobile networks \cite{Hu2003175, 579225}. Therefore, in
this work we study a theoretical model for ``distance frauds'' \cite{1039858}
in an abstract setting of distance-vector protocols, and aim to provide a
starting point to investigate security issues of other complex routing
protocols. 

The model is stated formally in Section \ref{sec:models}, but can be loosely
described as follows. For a graph $G$ in which vertices are agents, a subset
$S$ of agents are designated ``colluders'' and the rest are ``honest''. Honest
agents maintain a distance-vector recording their believed distances to all
other agents in the network, broadcasting this information to their immediate
neighbors in each round and updating their beliefs with the information
broadcasted to them.  When honest agents send or forward a message, they route
it to any neighbor that is closest to the message's recipient. Meanwhile,
colluding agents have knowledge about the entire graph and want to maximize the
number of messages that are routed through at least one member of $S$. They can
achieve this by lying in their broadcasts to their neighbors making it appear
they are closer to the message's destination than they actually are. As a
result, the honest agents will unknowingly forward messages to the colluding
agents along a potentially suboptimal path (See, for example, Figure
\ref{fig:rerout}).

While it would seem easy for the colluding agents to lie in order to intercept
messages, we impose the additional requirement that every message must
eventually reach its intended destination, as frequently dropped messages would
raise an alarm causing the colluders to be discovered. This makes strategy
design difficult for the colluding agents, as they must carefully balance lying
in order to attract messages while not overstating their proximity to the
recipient which would cause never-ending cycles. 

Solving this problem not only provides a tool for analyzing the susceptibility
of a network to endogenous information interception, but further informs
protocol designers of the vulnerability of honest agents naively following
protocols of this type. Our contributions are to analyze this model at many
levels from both theoretical and practical perspectives. In this article, we
accomplish the following:

\begin{itemize}
\item Formulate appropriate generalized combinatorial models for trust-based
communication in networks (Section \ref{sec:models});

\item Demonstrate that the case where each colluding agent may broadcast a
different piece of information to each of its neighbors, the non-uniform
broadcast model, reduces to the case where each individual colluding agent must
broadcast the same information to its neighbors, the uniform broadcast model
(Proposition \ref{prop:uniform-reduction});

\item Given a set of malicious agents $S$ where none of them are adjacent,
provide the optimal strategy regarding how the agents of $S$ should broadcast
(Theorem \ref{thm:optimal-separated});

\item Show that optimally choosing a set of colluding agents is NP-complete
(Section \ref{sec:combinatorial}); and 

\item Provide simulations demonstrating that, for various types of real-world
and synthetic networks, a very small portion of colluding agents, acting
strategically,  can in fact intercept a significant proportion of messages.
These results add a new perspective on the attack tolerance of scale-free
networks~\cite{AlbertJB00}. In addition to being vulnerable to connectivity
attacks by removing high degree nodes, they are vulnerable even to random
interception attacks (Section \ref{sec:simulations}).

\end{itemize} 

This paper is organized as follows. In Section~\ref{sec:related} we review
related work. In Section~\ref{sec:models} we define our models, and discuss
their complexity in Section 4. In Section~\ref{sec:strategies} we define our
strategies, and in Section~\ref{sec:separated} we prove the optimality of our
strategy (Theorem~\ref{thm:optimal-separated}). We generalize the strategy of
Theorem~\ref{thm:optimal-separated} to the general case of connected agents in
Section~\ref{sec:unseparated}. In Section~\ref{sec:simulations} we empirically
evaluate the quality of our strategies, and in Section~\ref{sec:conclusion} we
conclude with open problems.

\section{Related work} \label{sec:related}

There has been an extensive literature on secure designs of the distance-vector
routing protocol \cite{Hu2003175, 579225, Tao, Lopez} and on algorithms that
detect potential security issues \cite{1203939, 1258478, 1367222}. Our work
deviates from those lines in mainly two ways. First, as discussed in the
Introduction we study a theoretical model for ``distance frauds''
\cite{1039858} in an abstract setting of distance-vector protocols. Second,
rather than developing algorithms to detect malicious agents we aim to
understand the potential of groups of dishonest agents to divert traffic
through their infrastructure. Recent high-profile internet
outages~\cite{Stone08} have proven that traffic attraction or interception is
an important security issue of modern Internet, attracting large attention on
the threat models, incentive schemes, and secure alternatives for
BGP~\cite{ButlerFMR10,GoldbergSHR10,BallaniFZ07,NordstromD04,GoldbergHJRW08,LevinSZ08}.

Our work is in the spirit of~\cite{BallaniFZ07}, where the authors seek to
provide estimates on the amount of traffic that can be intercepted or hijacked
by colluding agents.  In contrast to their work, we want to focus on
information interception and we disallow ``black-holes'', a type of attack that
results in numerous dropped messages.  Consequently, identifying strategies
that promote information monitoring by maximizing traffic to the colluding
agents  becomes highly non-trivial.  In order to ease the theoretical analysis,
we move away from BGP and adopt the framework of distance-vector routing
protocol.  Furthermore, we assume that the honest agents are unsuspecting of
the possibly nefarious objectives of the colluding agents.  We prioritize
studying traffic interception over hijacking attacks with ``black holes'', as
in the latter case the honest agents would immediately realize that one of
their routing paths has been compromised.  Consequently, we believe our model
will serve as a fundamental tool in analyzing the more involved problem (not
discussed here) where honest agents actively seek out the colluding agents, and
the colluding agents must not only optimize the amount of traffic sent to them,
but also they must conceal their identities from the honest agents.

\section{Model and preliminaries} \label{sec:models}

For our study, we focus on a finite, unweighted, undirected graph $G = (V,E)$
with $V = \{1, \dots, n\}$ being the agents. We will write $i \sim j$ if
$\{i,j\} \in E$, i.e., nodes $i$ and $j$ are adjacent.

The model employs two types of agents represented by the vertices. First are
{\it honest} agents which follow an automated protocol. This protocol allows
the honest agents to learn about the distances between themselves and various
other agents over time. As a result, for any given destination, each honest
agent learns to which of its neighbors it should ``optimally'' route a message.
Second are {\it colluding agents} who make various choices within mild
constraints. Colluding agents act cooperatively with the goal of intercepting
as many messages as possible. They achieve this by broadcasting (potentially)
false distances thereby tricking the honest agents into routing messages along
a suboptimal path.

\subsection{What honest agents do} \label{sec:synchronization}

Each honest agent $i$ maintains an $n$-length vector, $(\rho^t(i,1),
\rho^t(i,2),\cdots,\rho^t(i,n))$, corresponding to all the believed distances
for agent $i$ between $i$ and every agent at time $t$.

Each agent, $i$, initially sets $\rho^0(i,j)$ to

\begin{equation} \label{uniformrho}
   \rho^0(i,j) =
    \begin{cases}
        0 & \text{if } i = j \\
        \infty & \text{otherwise.}
    \end{cases}
\end{equation}
 
While honest agents do not know the topology of the network, they can learn
about the distances between other nodes with a {\it synchronization protocol}.
At each step, each honest agent will receive broadcasts from its neighbors
(including both honest and colluding agents) on their believed distances and
update the information as follows.

\begin{equation}
   \rho^{t+1}(i,j) = \min \{ \rho^{t}(i,j), 1 + \min_{k \sim i} \rho^{t}(k,j) \} \label{uniformd}
\end{equation}

An example of the synchronization process is illustrated in Figure
\ref{fig:synchronization}. One can observe that the distance $\rho^t(i,j)$ is
nonincreasing with respect to $t$. Further, if all agents are honest the
$\rho^t(i,j)$ will converge in at most $n-1$ rounds. We will denote $\rho(i,j)$
the stationary distances in the end of the synchronization protocol. All our
analyses are based on the stationary distances $\rho(i,j)$.

Afterwards, each honest agent, $i$, establishes a {\it forwarding policy} which
is an array of $n$-sets, $\mathbf{b}_{i} = (b_{i,j})_{j=1}^n$ where $b_{i,j}$
is the set of neighbors of $i$ advertising the shortest path to $j$. That is,

\begin{equation}
{b}_{i, j} =  \text{argmin}_{k \sim i} ~ \rho(k,j) \label {uniformb} 
\end{equation}

where the {\it argmin} is a set of all such minimizing arguments. Finally, when
agent $i$ must forward a message whose destination is node $t$, it forwards to
a uniformly random node in $b_{i,t}$. The random choice is made independently
for every forwarding event.

Observe that if all agents are honest, the $\rho(x,y)$ is precisely the graph
distance from $x$ to $y$ which we denote as $d(x,y)$, and also for any
shortest-distance path $x= v_0, v_1, \ldots, v_{\ell} =  y$, we have $v_{i+1}
\in b_{v_i,y}$ for all $i=0, \ldots, \ell-1$. Conversely, for a node $w$ to
optimally route a message to $y$, $w$ simply forwards the message to any member
of $b_{w,y}$.

\begin{figure}
\centering
\includegraphics[width=1\textwidth]{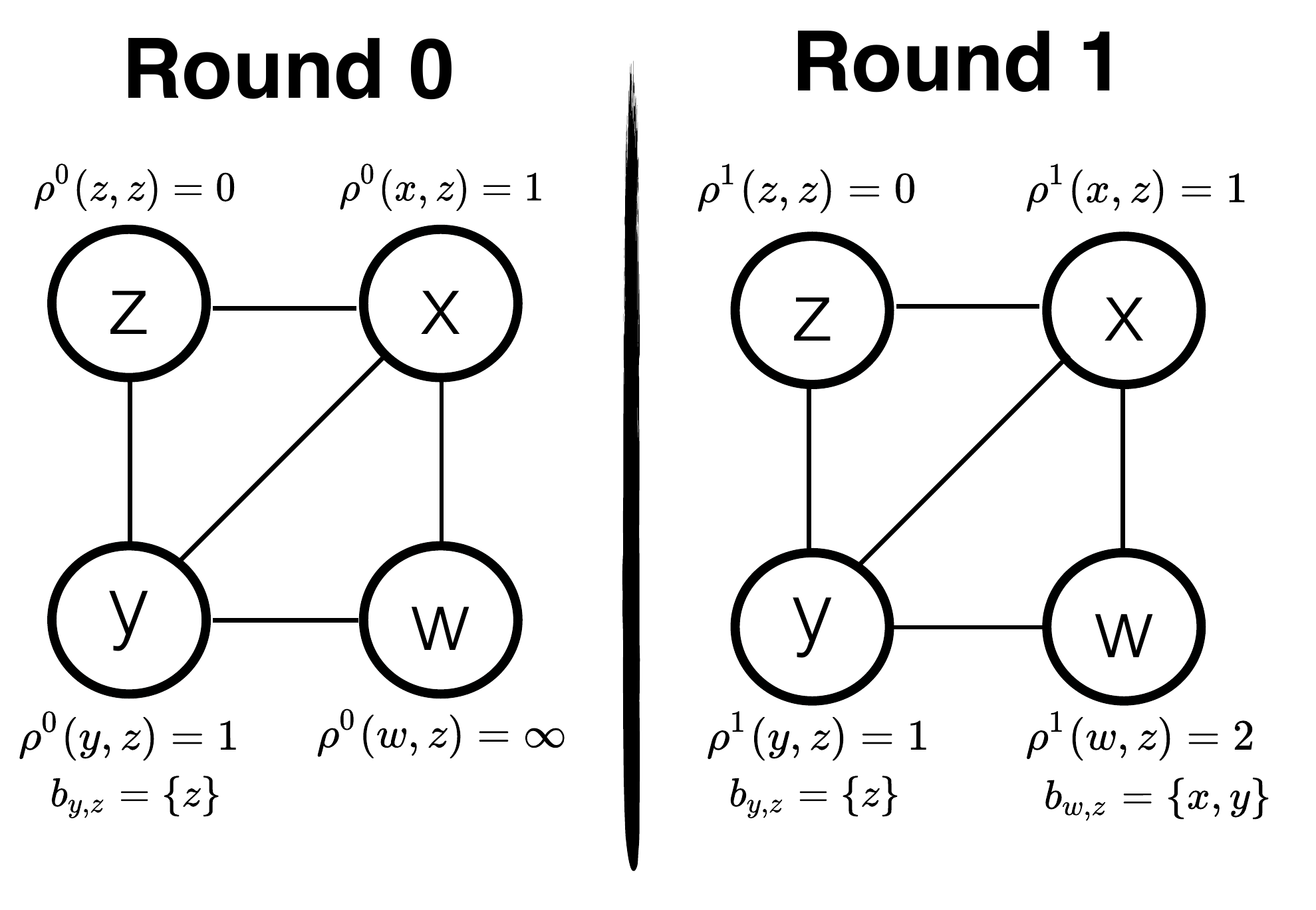}
\caption{(Left) Initially, nodes $x$ and $y$ know that they are neighbors to
node $z$ and have a perceived distance $\rho = 1$ to $z$. In contrast, node $w$
does not know its distance to $z$ and denotes it as $\infty$. (Right) After 1
round, nodes $x$ and $y$ both tell $w$, that they are distance 1 from $z$.
Hence, node $w$ updates is distance $\rho$ to $z$ to be 2 and stores which of
its neighbors ($x$ and $y$) are on the shortest path from $w$ to $z$.}
\label{fig:synchronization}
\end{figure}

\subsection{What colluding agents can do}
\label{sec:uniform-nonuniform-models}

The colluding agents work together and form a \emph{strategy} by choosing a
broadcasting distance $\rho^t(v,\cdot)$,  and a {\it forwarding policy}
$\mathbf{b}_{v}$ for each colluding agent $v$. The colluding agents do not
update their believed distances $\rho$ based on the broadcasting distances from
their neighbors; instead, they can broadcast any distances to gain traffic, and
our goal is to find the ``best'' value of $\rho(v,\cdot)$ for every colluding
agent $v$. In this case broadcasts will not change across rounds, which means
that agent $v$ will broadcast the same $\rho(v,\cdot)$  across all rounds, and
so we drop the superscript for time for colluding agents.  

During each step of the the synchronization protocol, each colluding agent will
broadcast its $\mathbf{\rho}$ to all of its neighbors. The honest neighbors
will perceive these distances to be real and update their believed distances
$\rho$, and $\mathbf{b}$ in equations \ref{uniformd} and \ref{uniformb}
accordingly, converging in $n-1$ steps. 

Allowing any strategy will potentially cause messages to route in a loop
without reaching its destination. For example, see Figure
\ref{fig:admissibility}. Hence, we restrict the space of admissible strategies
to be those that do not deterministically cause such loops. Further, agents
cannot falsely advertise to be destinations they are not. In the language of
BGP, we disallow ``black holes.'' 

To define this formally, we say a path from $s$ to $t$, $s= v_0, v_1, \ldots,
v_{\ell} = t$, is a {\it corresponding path} if $v_{i+1} \in b_{v_i,t}$ for
$i$ an honest agent. That is, a corresponding path has each honest agent
forward each message to a neighbor advertising a minimum distance where
colluding agents choose to whom they forward messages. A strategy is then
\emph{admissible} if there is a corresponding path between every pair of nodes.  

An admissible strategy is one for which every message will with high
probability over tie breaks by honest agents, eventually reach its destination.
Although it may travel along cycles during its routing, with overwhelming
probability the number of steps will be polynomial in the size of the graph.
This may be thought of as the weakest form of avoiding black holes, and it
allows our analysis to focus on ensuring the existence of \emph{some} optimal
path from the source to the target. Equivalently (up to high probability over
tie-breaks), we may assume that honest agents, when faced with a tie between an
option that will cycle and an option that will not cycle, will choose the latter to break
the tie.

It is worth noting that for every set $S$ the strategy where each agent in $S$
acts like an honest agent by broadcasting the true distance is admissible. We
call this the {\it honest strategy}.

\begin{figure*}
\centering
\includegraphics[width=\textwidth]{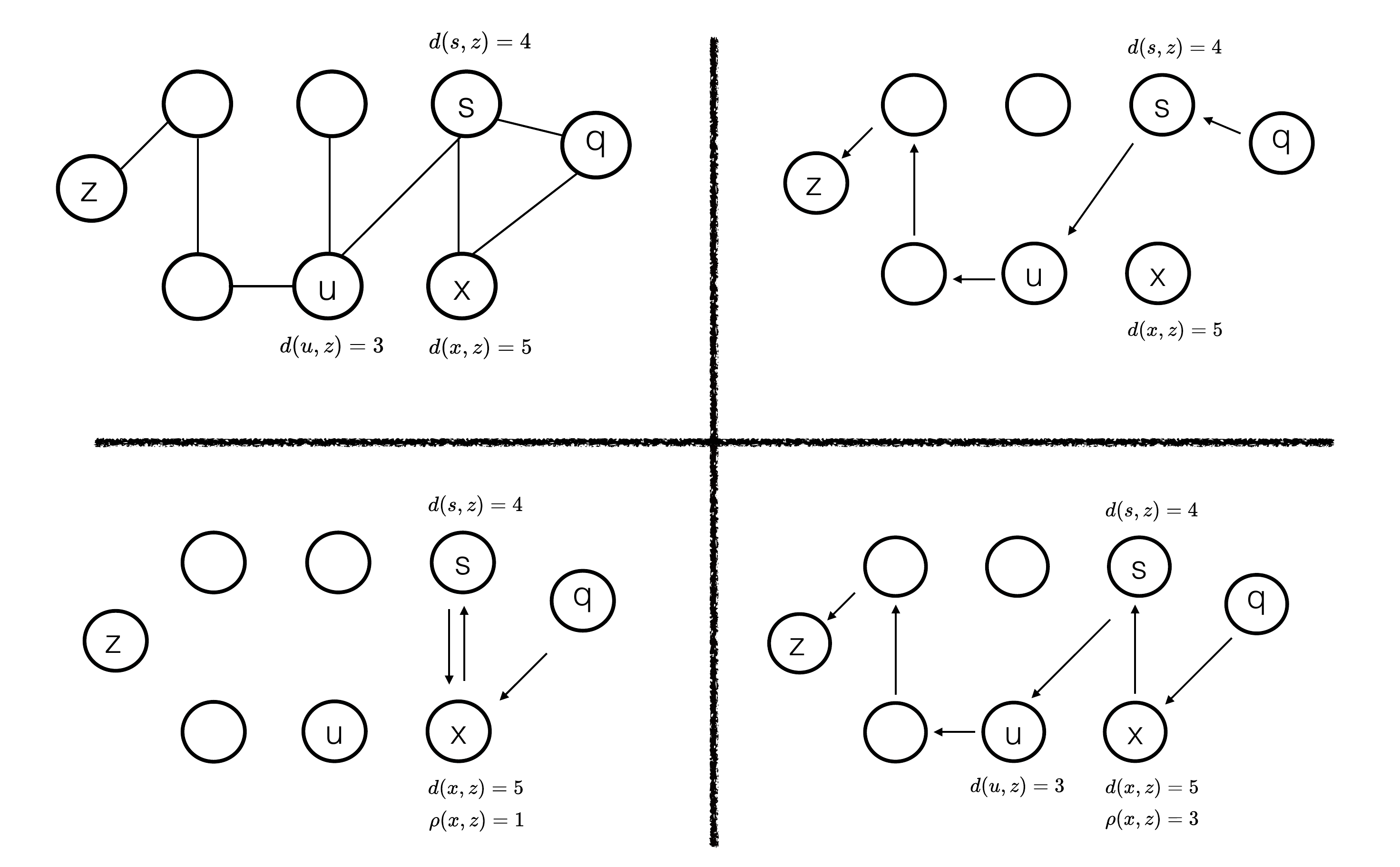}
\caption{For simplicity, in the figure above we drop the time subscript for
$\rho$ and assume the broadcast does not change. Consider an undirected graph
(top left) where the node $q$ wants to send a message to $z$.  If nodes are
broadcasting their true distances (top right), the induced directed graph
illustrates how the message travels from node $q$ to node $z$.  If node $x$, a
colluding node, overzealously broadcasts a distance that is too small to all of
its neighbors (bottom left), the message will never reach $z$ and hence is
inadmissible.  However, if $x$ broadcasts the same false distance (of 3) to all
of its neighbors as in the uniform model (bottom right), node $x$ captures the
traffic coming from $q$.  We construct a corresponding path by considering the
case where node $s$ forwards the message to $u$ (as $u$ and $x$ are
broadcasting the same distance), which demonstrates that the strategy is indeed
admissible.}
\label{fig:admissibility}
\end{figure*}

We analyze the effectiveness of a strategy by the number of node-pairs for
which the colluding agents necessarily intercept a message. In this sense, we
take the worst-case approach for the colluding agents. To make this formal,
given a set of colluding agents $S$ and admissible strategy $\theta$, let
$p_{S,\theta}$ be the proportion of pairs of nodes $\{i, j\}$ such that {\em
every} corresponding path between $i$ and $j$ passes through some node in $S$.
Further, a strategy is \emph{beneficial} if $p_{S,\theta} > p_{S,
\theta_{\textup{honest}}}$, where $\theta_{\textup{honest}}$ is the honest
strategy.

\section{Algorithmic Analysis}

In this section, we explore various algorithmic questions. First, we
demonstrate that optimally selecting a set of colluding agents of a fixed size
is NP-hard. We achieve this by reducing the problem to vertex covering.
Additionally, we give examples where the optimal effectiveness of a set of
colluding agents $S$ is not submodular. Finally, we prove that a variant model
where the colluding agents $S$ may broadcast different, potentially false,
messages to its neighbors reduces to the case where the broadcast is uniform.

\subsection{NP-hardness and submodularity} \label{sec:combinatorial}

We investigate the algorithmic properties of optimally choosing a set of
colluding agents by considering the case when the agents use the honest
strategy, $\theta_{\textup{honest}}$. In which case, how do we pick which nodes
should collude? In fact the problem of optimally selecting nodes reduces to the
following:

\begin{problem}[$p$-Shortest Path Dominating Set] \label{prob:spds}
Given a graph $G = (V,E)$ an integer $k \geq 1$ and $p \in [0,1]$, is there a
set $S \subset V$ of size $k$ so that at least $p \binom{n}{2}$ shortest paths
pass through $S$?
\end{problem}

We abbreviate this by SPDS$(p,k)$ and say a vertex $v$ \emph{covers} a path if
it lies on the path. There are two natural optimization problems associated
with SPDS. The first, MAX-SPDS$(-,k)$, is to maximize the $p$ achieved over all
sets of size $k$. The second, MIN-SPDS$(p,-)$, minimizes $k$ while attaining a
prespecified $p$. SPDS$(p,k)$ is trivially NP-hard because SPDS$(1, k)$ is
VERTEX-COVER.
 
By a standard argument, the function $f: 2^V \to [0,1]$ mapping $S$ to the
proportion of shortest paths covered by $S$ is submodular\footnote{Recall a set
function $f:2^X \to \mathbb{R}$ is \emph{submodular} if for every $S, T \subset
X$ with $S \subset T$ and for every $x \in X \setminus T$, the marginal gain
$f(S \cup \{ x \}) - f(S) \geq f(T \cup \{ x \}) - f(T)$} and monotone.  Hence,
by a classic theorem of Nemhauser, Wolsey, and Fisher~\cite{NemhauserWF78}, the
greedy algorithm provides a $(1-1/e)$-approximation algorithm for
MAX-SPDS$(-,k)$. A slight variant of the greedy algorithm presented by
Wolsey~\cite{Wolsey82} achieves a $p$-proportion of shortest paths with a set
$S$ of size

$$ |S| = \left [ 1 + \log \max_{v \in V} f(\{ v \}) \right ] OPT
       = \left [ 1 + O(\log(n)) \right ] OPT$$
where OPT is the size of the smallest set covering a $p$-proportion of shortest
paths. 

A similar reduction from VERTEX-COVER shows that the problem of picking
colluding nodes \emph{and} a good strategy for lying is also NP-hard. Moreover,
the collusion problem is not submodular. We prove this below, and as such, we
will henceforth focus on the problem of determining the optimal strategy for a
given set of colluding agents.

\begin{proposition} \label{prop:not-submodular}

Fix $\theta(S)$ mapping subsets of vertices to optimal strategies. Define $f:
2^{V(G)} \to \mathbb{R}$ by letting $f(S)$ be the number of shortest paths
passing through $S$ when $S$ uses the strategy $\theta(S)$. Then $f$ is not
submodular.

\end{proposition}

\begin{proof}

Let $G = K_{m,2}$ be the complete bipartite graph on parts $X,Y$ with $X = \{
p,q \}, |Y|=m$. Let $S = \{ \}, T = \{ p \}$. Then adding $q$ to $T$ captures
every path in $G$. Adding $q$ to $S$ captures only the messages sent to and
from $q$, because other messages are routed via $p$. In particular,

\begin{align*}
   f(S \cup \{ q \}) - f(S) &= 2(m+1) = O(m) \\ 
   f(T \cup \{ q \}) - f(T) &= 2 \binom{m+2}{2} - 2(m+1) = O(m^2),
\end{align*}
disproving submodularity. 
\end{proof}

There are also interesting examples of this that rely on lies in a nontrivial
way, for example when $G = C_n$ for sufficiently large $n$. In this case a
single node has little advantage, but adding a second adjacent node allows for
an interesting collusion. One of the two colluders may broadcast 1 for most
targets, and forward incoming messages to the its neighboring colluder who in
turn forwards along $C_n$ to the target. We discuss the broadcast bounds for
this case in Section~\ref{sec:unseparated}.

\subsection{Non-uniform broadcasting and a reduction} \label{sec:reduction}

Previously, we required each colluding agent to broadcast its false distances
uniformly to all of its neighbors. However, realistically, a colluding agent
could make varying broadcasts to different neighbors. We refer to this as the
{\it nonuniform model} and the previous as the {\it uniform model}.

Formally, the only difference between the models is that for the nonuniform
model, the colluding agents must establish a more precise broadcast strategy,
$\rho(v,t,i)$, for each colluding agent $v$ and each neighbor $i$. 

Nonuniform lies might appear at first glance to provide a substantial increase
in power. In particular, there are many more strategies, and it seems easier to
accidentally introduce a large cycle, so it seems \emph{computationally} harder
to find an optimal solution. However, we show this is not the case. Finding the
optimal nonuniform strategy for a fixed subset is reduces to finding the
optimal uniform strategy. Because the uniform model is also a special case of
the nonuniform model, this shows the two problems are computationally
equivalent. This justifies a focused study of the uniform model.

Consider the following decision problem:

\begin{problem}[UNIFORM-SUBSET-MONITORING]

Given a graph $G = (V,E)$, a subset $S \subset V$ and a $p \in [0,1]$, is there
an admissible strategy $\theta$ for $S$ such that in the uniform broadcast
model $p_{S, \theta} \geq p$?

\end{problem} 

 Similarly for nonuniform lies, we analogously define
NONUNIFORM-SUBSET-MONITORING. We will prove that the nonuniform case reduces to
the uniform case with a small blowup construction.

\begin{proposition} \label{prop:uniform-reduction}

NONUNIFORM-SUBSET-MONITORING reduces to {{UNIFORM-SUBSET-MONITORING}}.

\end{proposition}

\begin{proof}

Given a graph $G$, a subset $S$, and a fraction $p$ for the nonuniform model,
we produce a new graph $G'$, a subset $S'$, and a fraction $p'$ for the uniform
model as follows.  (See figure 4 for an illustration). For simplicity we will
prove the case where $G$ is $D$-regular. Start by setting $G' = G$. For each
edge $e = (u,v)$ where $u \in S$, subdivide $e$ in $G'$ with a new vertex
$w_e$. Also for each such edge, add $u, w_e$ to $S'$.  Finally, set 

\[ 
   p' = \frac{p\binom{|V|}{2} + \binom{|S|D}{2} + |S||V|}{\binom{|V| +
   |S|D}{2}}.
\]

Suppose there is a strategy $\theta$ for $S,G$ achieving a $p$ fraction in the
nonuniform model. We'll convert $\theta$ into a strategy $\theta'$ for $S'$.
Whenever a colluding agent $u \in S$ would broadcast $\rho$ to a neighbor $v$,
we have the agent $w_{(u,v)}$ uniformly broadcast $\rho$.  And whenever a
message goes to $w_{(u,v)}$ with some other destination, $w_{(u,v)}$ forwards
it through $u$, who in turn forwards it to the $w_{(u,v')}$ corresponding to
the same $v'$ that $u$ would forward to in the nonuniform setting. This
simulates $\theta$, and hence achieves the same $p\binom{|V|}{2}$ paths in $G$;
the formula for $p'$ simply counts the paths introduced by the new vertices
(all of which include a colluding agent). So if $p_{S,\theta} \geq p$ in $G$,
$p_{S', \theta'} \geq p'$ in $G'$.

Conversely, we can collapse any uniform strategy for $S'$ into a nonuniform one
for $S$ by contracting all the newly added edges in $G'$ and combining their
broadcasts in the obvious way. The case where $G$ is irregular is analogous,
and it's clear the appropriate $p'$ can be efficiently computed.

\end{proof}

\begin{figure}
\centering
\includegraphics[width=\textwidth]{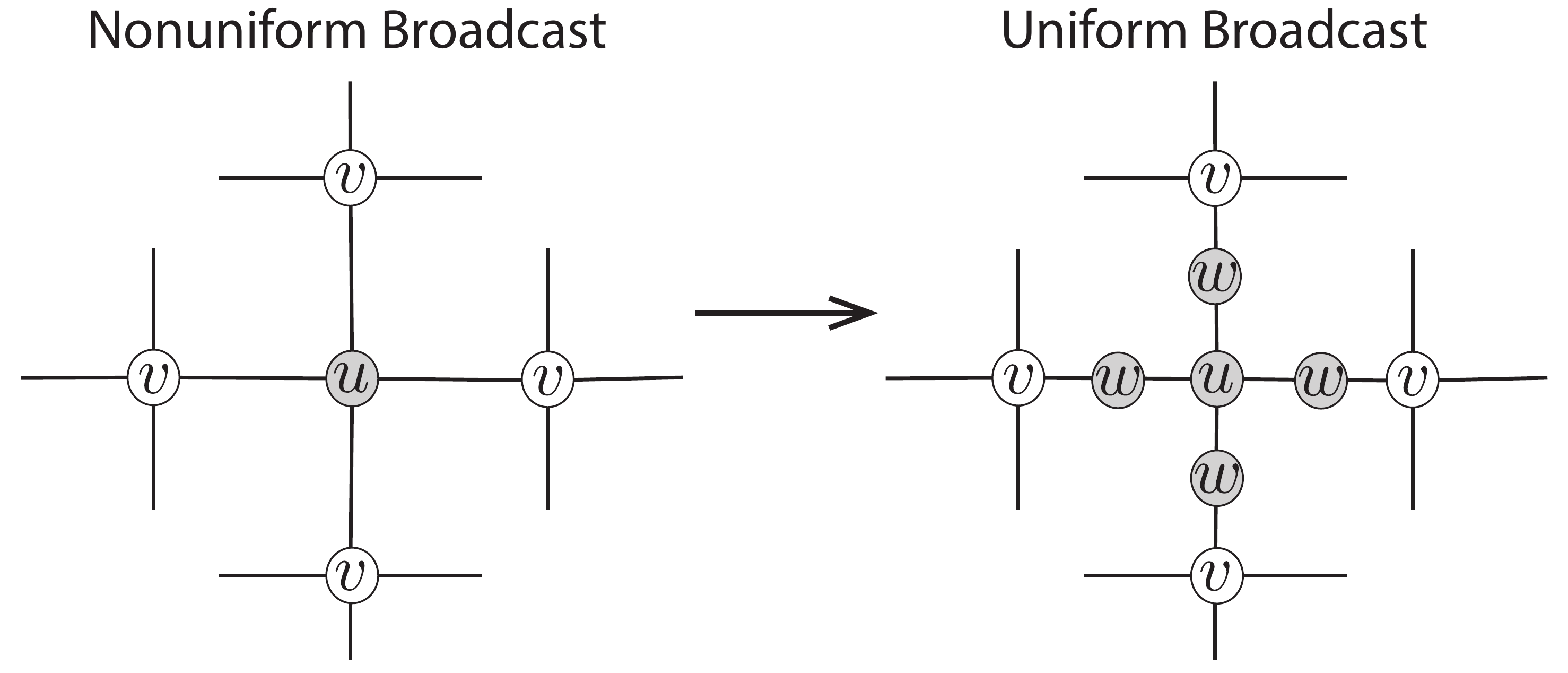}
\caption{An illustration for Proposition \ref{prop:uniform-reduction}. Given a
graph $G$, a set of colluding agents $S$ (shaded circles), and a fraction $p$
for the nonuniform model (left panel), we produce a new graph $G'$, a new set
of colluding agents $S'$ (shaded circles), and a fraction $p'$ for the uniform
model (right panel) by inserting a new colluding agent to every edge connecting
a colluding agent in $S$ and a honest node. We show in the proposition that a
nonuniform broadcasting strategy in the left is equivalent to an uniform
broadcasting strategy in the right.} \label{fig:nonuniform-cycle}
\end{figure}

\section{Strategies for the uniform model} \label{sec:strategies}

We now turn to a detailed study of the uniform broadcast model. In this
section, we define the optimal strategy in Section~\ref{sec:separated} and
prove that it is optimal. Minimizing the distances broadcasted by the colluding
agents optimizes the traffic captured by the agents; therefore, to prove the
optimality of the strategy we show that it achieves the minimal distances that
could be broadcasted by the colluding agents without causing ``black holes''.
Our main result is the following:

\begin{theorem} \label{thm:optimal-separated}

Let $G = (V,E)$ be a graph and $S \subset V$ a fixed subset of vertices such
that the actual distance $d(x, y) \geq 2$ for all $x,y \in S$. Then the
strategy $\rho^*$ defined in Section~\ref{sec:separated} intercepts an optimal
fraction of traffic in $G$.

\end{theorem}

The proof consists of two parts. After defining the broadcast policy $\rho^*$
in Section~\ref{sec:separated} and the associated forwarding policy, we first
prove in Proposition~\ref{prop:rhostar-admissible} that the strategy is
admissible, and then in Proposition~\ref{prop:rhostar-lower-bound} that
$\rho^*$ is the minimal broadcast for any admissible strategy.

\subsection{A single agent} \label{sec:single-agent}

We start by characterizing the case of a single colluding agent. This case is
useful because it forms a ``base case'' for our more complicated strategies
later. In this case, the best strategy is simple: lie exactly two less than
your true distance to a target. We show that this is guaranteed not to cause
cycles while any larger lie causes a cycle of length 2. 

Before we state the theorem below, we remind readers that $d(x,t)$ is the true
graph distance from node $x$ to node $t$.

\begin{theorem} \label{thm:single-agent}

Let $x$ be a colluding node and $t$ be a target node whose true distance in $G$
is $d(x,t) = k$. Suppose that $x$ broadcasts $\rho(x,t) = k'$. Then this
strategy is admissible and beneficial if and only if one of the following
conditions hold.

\begin{enumerate}
   \item $k' = k - 2$ and there is a neighbor $z$ of $x$ with $d(z,t) = k = k'
+ 2$.
   \item $k-2 \leq k' \leq k-1$ and there is a neighbor $z$ of $x$ with $d(z,t)
= k + 1$, and there is a shortest path from $z$ to $t$ that does not pass
through $x$ (before $x$'s lie).
\end{enumerate}

\end{theorem}

\begin{proof}

Note that $d(x,t) = k$ if and only if the closest neighbor $y$ of $x$ to $t$
has distance $d(y,t) = k-1$. The strategy for $x$ will be to route all messages
to $t$ through $y$.

For one direction, suppose one of the above conditions holds and let $z$ be a
neighbor of $x$ satisfying the desired property. Then $z$ will send messages to
$t$ through $x$, which $x$ can forward through $y$. We further claim that no
message forwarded through $y$ to $t$ will ever come back to $x$. Call $y, v_2,
v_3, \dots$ the vertices on the route taken by the sent message after passing
through $x$. We construct a corresponding path by requiring that $v_2 \neq x$
(as $d(y,t)=k-1$ ensures that $y$ has a neighbor $v_{2}$ such that
$d(v_{2},t)=k-2$) and it follows that the perceived distances $\rho(v_i, t)$
along the message's path from $t$ are monotonically decreasing in $i$. This
follows from the fact that $x$ broadcasts the same lie to all its neighbors. In
particular, $\rho(x,t) = \rho(v_2, t)$ and so for all $j \geq 2$, $v_j$ will
always have a closer neighbor than $x$.

For the converse, suppose the strategy is admissible and beneficial. First, $x$
cannot broadcast $\rho(x,t) < k-2$, or else $y$ (and by minimality all
neighbors of $x$) will route messages to $t$ through $x$, causing a cycle. If
$k' \geq k$, then no neighbor of $x$ would change its behavior, contradicting
beneficialness. This implies the conditions on $k'$ in (1) and (2). Moreover,
beneficialness implies $x$ has a neighbor $z$ that now forwards messages
through $x$, implying its new perceived distance is $\rho(z,t) = k' + 1$. By
the fact that $d(x,t) = k$, every other neighbor $z$ of $x$ has $k-1 \leq
d(z,t) \leq k+1$. If all neighbors have distance $k-1$ then no node is forced
to send messages to $x$ and implies the strategy is not beneficial, so let $z$
be a neighbor with $d(z,t) \geq k$.

If $d(z,t) = k$ then the shortest path from $z$ to $t$ already does not pass
through $x$ and we must choose $k' = k-2$ to change $z$'s behavior. On the
other hand, if $d(z,t) = k+1$ but has no other shortest path to $t$ except
through $x$, then lying is not beneficial. If $z$ has another
path to $t$ then setting $k' = k-1$ breaks the tie.
\end{proof}

One can simplify the above lemma by noting that setting $k' = k-2$ is always
nondetrimental, and this is the optimal nondetrimental lie. So if an agent has
incentive to lie, it may as well lie as much as possible. This motivates the
following corollary.

\begin{corollary} \label{cor:single-agent}

Let $x$ be a single lying agent in $G$ in the uniform local broadcast model. An
optimal admissible strategy for a single lying agent $x$ is to broadcast
$\rho(x,t) = \max(1, d(x,t) - 2)$ for all $t \in V(G)$.  

\end{corollary}

\begin{proof}
In the proof of Theorem~\ref{thm:single-agent}, we showed that a vertex $x$
lying in this way cannot produce any cycles of length 2, since it does not
alter the behavior of the neighbor through which $x$ routes messages to $t$. It
remains to show that there are no longer cycles.

Suppose to the contrary that when $s$ tries to send a message to $t$, there is
a cycle $v_1, v_2, \dots, v_m, v_{m+1} = v_1$, with $m \geq 3$. Let $i$ be the
index of a vertex on the cycle which minimizes the true distance $d(v_i, t)$.
Call this distance $a$, and note that $v_i$ is not a colluding agent (or else
it could correctly forward messages so as to break the cycle). Because there is
a cycle, $v_{i+1}$ is broadcasting $\rho(v_{i+1}, t) \leq a-2$, but $d(v_{i+1},
t) \geq a$. And since $v_{i+1}$ forwards to $v_{i+2}$, we have $\rho(v_{i+2},
t) \leq a-3$ while $d(v_{i+2}, t) \geq a$. We claim this is a contradiction: a
colluding agent lies by exactly two less than the truth, and so $v_{i+2}$
cannot be a colluding agent. But the effect of a colluding agent's lie does not
change the perceived distances of any vertex in $G$ by more than two. This
shows the claimed contradiction.
\end{proof}

The same algorithm can be jointly and independently used by multiple colluding
agents in the uniform broadcast model. We make this rigorous with the following
proposition.

\begin{proposition} \label{prop:independent-agents}

If any set of colluding agents lie independently according to
Corollary~\ref{cor:single-agent}, then their joint strategy is admissible.

\end{proposition}

\begin{proof}

In the proof of Theorem~\ref{thm:single-agent}, we showed that a vertex $x$
lying in this way cannot produce any cycles of length 2, since it does not
alter the behavior of the neighbor through which $x$ routes messages to $t$. It
remains to show that there are no longer cycles.

Suppose to the contrary that when $s$ tries to send a message to $t$, there is
a cycle $v_1, v_2, \dots, v_m, v_{m+1} = v_1$, with $m \geq 3$. Let $i$ be the
index of a vertex on the cycle which minimizes the true distance $d(v_i, t)$.
Call this distance $a$, and note that $v_i$ is not a colluding agent (or else
it could correctly forward messages so as to break the cycle). Because there is
a cycle, $v_{i+1}$ is broadcasting $\rho(v_{i+1}, t) \leq a-2$, but $d(v_{i+1},
t) \geq a$. And since $v_{i+1}$ forwards to $v_{i+2}$, we have $\rho(v_{i+2},
t) \leq a-3$ while $d(v_{i+2}, t) \geq a$. We claim this is a contradiction: a
colluding agent lies by exactly two less than the truth, and so $v_{i+2}$
cannot be a colluding agent. But the effect of a colluding agent's lie does not
change the perceived distances of any vertex in $G$ by more than two. This
shows the claimed contradiction.
\end{proof}

\subsection{Separated agents}\label{sec:separated}

We now turn to the case of multiple colluding agents. By
Proposition~\ref{prop:uniform-reduction}, we know that allowing neighboring
colluding agents introduces the ability for nonuniform broadcasts. So we
characterize the alternative where no two colluding agents are adjacent. The
optimal strategy we define generalizes to a nontrivial admissible strategy for
the general case in Section~\ref{sec:unseparated}.

For a set $X$ of integers and an integer $j$, define $S_X(j)$ to be the set of all
permutations of $j$ elements from $X$. In this section $C = \{ x_1, \dots, x_k
\}$ will denote the set of colluding nodes, and no two are adjacent in $G$.

\begin{definition}
The \emph{$j$-th colluding distance} between two colluding agents $x$ and $y$
is defined as

\[
   d_j(x, y)= \min_{\substack{\sigma \in S_C(j) \\ \sigma(1) = x \\ \sigma(j) =
y}}
         \sum_{i=1}^{j-1} d(\sigma(i),\sigma(i+1)).
\]

In other words, it is the length of the shortest path from $x$ to $y$ that
contains $j$ colluding nodes. Call any path minimizing this quantity a
\emph{$j$-th colluding path}. If no such path exists, call $d_j(x,y) = \infty$
by convention.  \end{definition}

We will consider $j$-th colluding paths directed from $x$ to $y$ when
appropriate. Now given a set of colluding nodes and a target vertex $t$, we
want to identify the strategy that minimizes $\rho(-,t)$ for all of our
colluders. We start by defining a candidate strategy $\rho^*$, observe that it
is admissible, and then prove it is indeed a lower bound on admissible
strategies.

\begin{definition} \label{def:colluding-distance}
Let $\rho'(x,t) = \max(d(x,t) - 2, 1)$. Let $\rho''(x, t)$ be defined as

\[
   \min_{1 \leq i,j \leq k} \left [ d_j(x,x_i)-2(j-1) + \rho'(x_i, t) \right ].
\]

Then define the strategy $\rho^*(x,t)=\min(\rho'(x,t), \rho''(x, t))$.

\end{definition}

To give some intuition, this strategy takes the minimum of
Corollary~\ref{cor:single-agent} and the best $j$-th colluding path (where the
end of that path uses Corollary~\ref{cor:single-agent} to get to $t$).

Definition~\ref{def:colluding-distance} is useful in the following scenario
depicted in Figure~\ref{fig:colluding-distance}. Suppose $(s, x_1, y_2, x_2,
\dots, y_j, x_j, t)$ is a path of length $2j + 1$, where $t$ is the target of a
message sent by $s$ and $x_i$ are colluding agents.  Then every $x_i$ may
broadcast $\rho^*(x_i, t) = 1$, and 
the honest agents will forward along the path toward $t$.

\begin{figure}[thb]
\centering
\includegraphics[width=\textwidth]{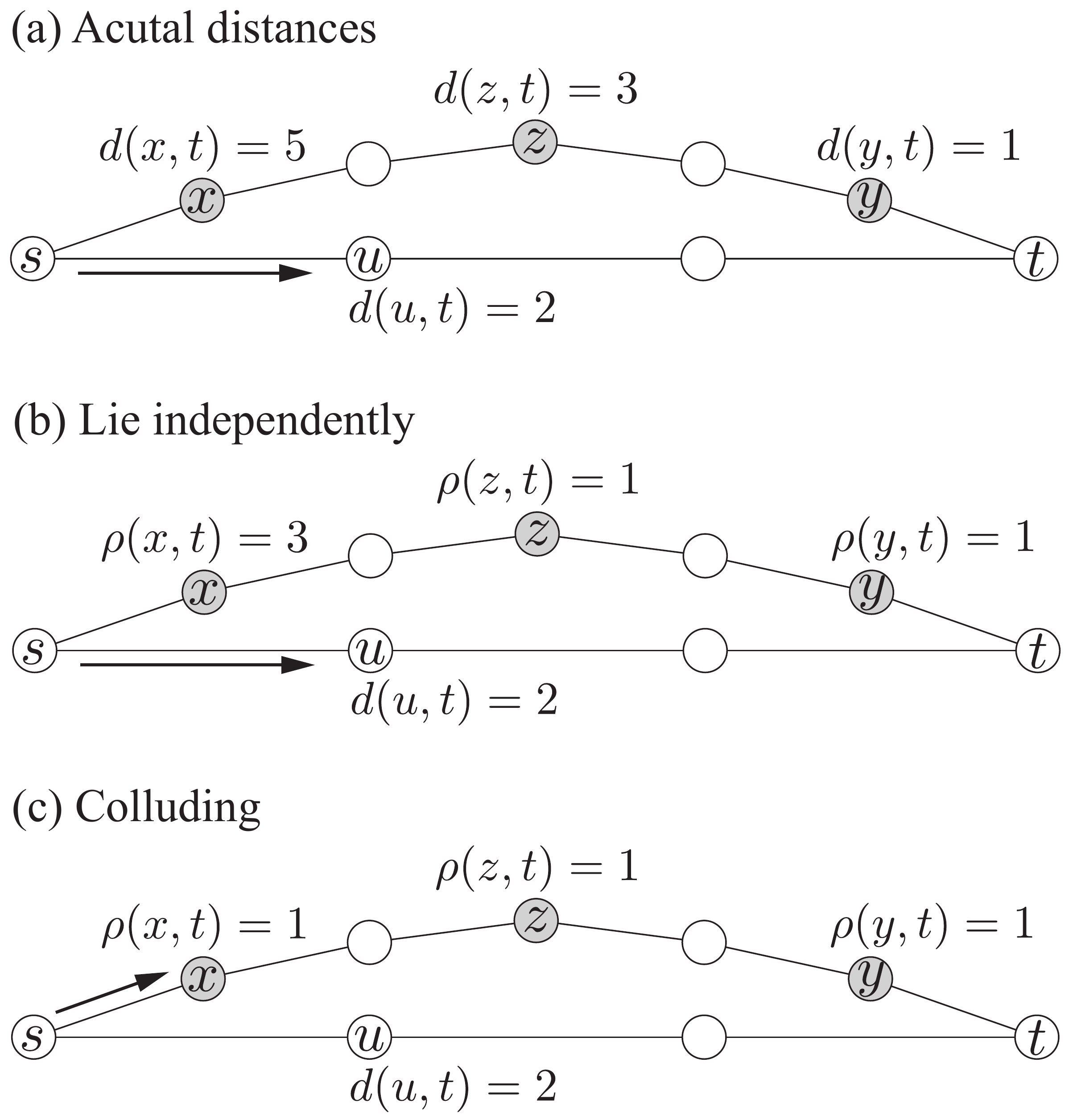}
\caption{An example of strategies for colluding agents $x$,$y$ and $z$ (shaded
circles). Suppose that $s$ wants to send a message to $t$. (a) If no lies
involved, $s$ would send the message to $u$ who is closest to $t$. (b) When
agents $x$,$y$ and $z$ do not collude, the best strategy for $x$ without
causing ``black holes'' is to broadcast $\rho(x,t) = d(x,t) - 2=3$. $s$ would
still send the message to $u$. (c) When taking into account other colluding
agents, $x$ may broadcast $\rho(x,t) = 1$ according to the optimal strategy
defined in \ref{def:colluding-distance}, and deceive $s$ to send the message
through $x$.  } \label{fig:colluding-distance}
\end{figure}

We call $x$ \emph{proper for $t$} if it (strictly) minimizes $\rho^*(x,t)$ via
the $j$-th colluding distance for some $j \geq 2$ and \emph{improper}
otherwise, and we call a $j$-th colluding path realizing this minimization a
\emph{witness} for the properness of $x$. We define the \emph{forwarding
number} of a $j$-th colluding path to be $j$, and define the forwarding number
of a vertex $v$ to be the smallest forwarding number of any $j$-th colluding
path minimizing $\rho^*(v,t)$ starting at $v$. Note that a vertex with a
forwarding number 1 is by definition improper, and that all of these
definitions depend on the choice of $t$.

Improper colluding nodes forward as in Corollary~\ref{cor:single-agent}. Proper
colluding agents pick a $j$-th colluding path which minimizes their broadcast.
We now prove that $\rho^*$ is an admissible strategy. But first, we prove that
$j$-th colluding paths can be extended in a nice way.

\begin{proposition} \label{prop:inductive-path}

Let $x,y$ be colluding vertices in $G$. Then $\rho^*(x,t) \leq d(x,y) - 2 +
\rho^*(y,t)$.

\end{proposition}

\begin{proof}

If $y$ is improper, then $\rho^*(x,t)$ has a 2-th colluding path and trivially
$\rho^*(x,t) \leq d(x,y) - 2 + \rho'(y,t)$ so we are done.

So suppose $y$ is proper with forwarding number $j$. By definition, there is a
$y'$ such that $\rho^*(y,t) = d_j(y,y') - 2(j-1) + \rho'(y',t)$. Call the
witness path $\sigma$. Then there is a corresponding path $\sigma'$ for $x$
constructed by prepending a path from $x$ to $y$ to $\sigma$. This path is some
$j'$-th colluding path for $j' > j$ whose cost is an upper bound on
$d_{j'}(x,y')$.  Hence,

\begin{align*}
   \rho^*(x,t) &\leq d_{j'-j}(x,y) - 2(j' - j) + d_j(y,y') - 2(j-1) +
\rho'(y',t) \\
               &\leq d(x,y) - 2 + \rho^*(y,t),
\end{align*}

as desired.
\end{proof}

In fact, if $x$ is proper, $\rho^*(x,t)$ is minimized by computing $d(x,y) - 2
+ \rho^*(y,t)$ for some colluding agent $y$. It will have the property that
there is a path from $x$ to $y$ that passes through no other colluding agents.
Proposition~\ref{prop:inductive-path} trivially extends to non-colluding agents
$y$, giving the following corollary. Note here we use $\rho$ to denote the
broadcast (honest or lie) of any agent.

\begin{corollary} \label{cor:rhostar-bound}

If all colluding agents are following $\rho^*$, then for all $x \in C, y \in
V(G)$, $\rho^*(x,t) \leq d(x,y) - 2 + \rho(y,t)$.

\end{corollary}

Another simple consequence of Proposition~\ref{prop:inductive-path} is that the
forwarding number decreases along minimal $j$-th colluding paths.

\begin{proposition} \label{prop:inductive-forwarding-number}

Suppose $x$ is a proper colluding agent with respect to $\rho^*(x,t)$, that $x$
has forwarding number $j$, and that $\sigma$ is a witnessing $j$-th colluding
path. Let $x'$ be the first colluding agent on $\sigma$ after $x$. Then $x'$
has forwarding number $j-1$.  

\end{proposition}

\begin{proof}

The same technique from the proof of Proposition~\ref{prop:inductive-path}
shows that whatever the forwarding number of $x'$ is, we can prepend a path from
$x$ to $x'$ to get a path with forwarding number $j+1$.
\end{proof}

In particular, the $i$-th visited colluding agent on a witness for
$\rho^*(x,t)$ of forwarding number $j$ has forwarding number exactly $j - i$,
and the end of the path is an improper colluding agent.

At this point one might expect some sensible extension of the pair of
(forwarding number, broadcasted distance) to honest agents would produce a
potential function that is monotonically decreasing along the message path and
zero at the target. Indeed, a version of this is true when the colluding agents
are separated by distance at least three, with ready counterexamples for
distance two. Still, we present a different argument that $\rho^*$ is also
admissible when the agents are distance two apart.
 
\begin{proposition} \label{prop:rhostar-admissible}

Let $C = \{ x_1, \dots, x_k \}$, and suppose that $d(x_i, x_j) \geq 2$ for all
$x_i, x_j \in C$. Then $\rho^*$ is admissible.

\end{proposition}

\begin{proof}

Let $s,t$ be arbitrary vertices in $G$, and let $L = (y_1, y_2, \dots, y_m)$ be
a cycle in the path of a message sent from $s$ to $t$ (possibly infinite and
repeating). If the $y_i$ are all honest or improper agents we are reduced to
the case of Proposition~\ref{prop:independent-agents}. So some of the $y_i$
must be proper colluding agents. 

Without loss of generality suppose $y_1$ is a colluding agent, and let $p$ be
the minimal colluding path it forwards along, extended to the target $t$.  Let
$y_j$ be the last vertex on $L$ that is not also on $p$. The claim is that
$y_j$'s decision to forward along $L$ or $p$ is a tie break. This proves the
proposition because we can then construct a corresponding path. 

Let $z$ be the vertex following $y_j$ on $p$, and suppose to the contrary
$\rho(y_{j+1}, t) < \rho(z,t)$. Let $x$ be the last colluding agent on $p$
before $y_j$. Let $x'$ be the first colluding agent on $p$ after $y_j$ (it may
be the case that $x' = z$). If $x$ is the last colluding agent on $p$, then let
$x'=t$ and the proof proceeds similarly. First we expand $\rho^*(x,t) + 2$
along $p$.

\begin{align*}
   \rho^*(x,t) + 2 &= d(x,x') + \rho(x',t) \\  
                   &= d(x,y_j) + 1 + d(z,x') + \rho(x',t) \\
                   &= d(x,y_j) + 1 + \rho(z,t)
\end{align*}

We now bound $\rho^*(x,t) + 2$ along $L$ using
Corollary~\ref{cor:rhostar-bound}. 

\begin{align*}
   \rho^*(x,t) + 2 &\leq d(x, y_j) + 1 + \rho(y_{j+1}, t) \\  
                   &< d(x, y_j) + 1 + \rho(z,t) \\ 
                   &= \rho^*(x,t) + 2
\end{align*}

A contradiction.
\end{proof}

Next we prove that in the separated setting $\rho^*$ is a lower bound on
admissible broadcasts. 

\begin{proposition} \label{prop:rhostar-lower-bound}

Any colluding agent broadcasting $\rho(x,t) < \rho^*(x,t)$ necessarily causes a
cycle.

\end{proposition}

\begin{proof}

Suppose to the contrary some colluding agent $x$ broadcasts $\rho(x,t) <
\rho^*(x,t)$. We will show that all neighbors of $x$ forward to $t$ through
$x$, necessarily causing cycle of length 2. Fix any neighbor $z$ and suppose to
the contrary that there is a neighbor $y \neq x$ of $z$ with $\rho(y,t) \leq
\rho(x,t)$. But then $\rho(x,t) < \rho^*(x,t) \leq d(x,y) - 2 + \rho(y,t) =
\rho(y,t)$ by Corollary~\ref{cor:rhostar-bound}, a contradiction.
\end{proof}

This completes the proof of Theorem~\ref{thm:optimal-separated}.

Finally, $\rho^*$ can be efficiently computed. The idea is to grow a search
tree of colluding agents from $t$, noting that the value of $\rho^*$ for a new
vertex is minimized by using some set of previously visited nodes. More
rigorously, for each target $t \in V(G)$ run the following procedure. Set $S =
\{ t \}$.  While $C \not \subset S$ is missing some colluding agent, take any
colluding agent with minimal distance to $S$ (true distance in $G$), and
calculate the value of $\rho^*(x,t)$ as $\rho^*(x,t) = \min_{y \in S} d(x,y) -
2 + \rho(y,t)$. Then add $x$ to $S$ and continue. Using the same arguments used
previously, it is easy to see that this will compute the correct value of
$\rho^*$ for every colluding agent and every target. Moreover, one can
construct the corresponding $j$-th colluding paths during this process. We
provide some example simulations of using this strategy on synthetic and
real-world networks in Section~\ref{sec:simulations}.

\subsection{Adjacent colluding agents} \label{sec:unseparated}

In this section we extend the strategy from Section~\ref{sec:separated} to the
setting where colluding agents may be adjacent in the network. We show this
generalization is not optimal, and instead give a family of strategies, one of
which must be optimal.

Before we state our theorems, we describe another connection between the
uniform and non-uniform models from Section~\ref{sec:models}, that we can
transform an instance of the uniform model into an instance of the nonuniform
model in which colluding agents are separated.  Specifically, one can take the
quotient $G/\sim$ of the graph $G$ by declaring two colluding agents to be
equivalent if they are in the same connected component in the induced subgraph
of colluding agents. Uniform strategies translate into nonuniform ones as
follows. If $A$ is a connected component of colluding agents collapsed to $v_A$
with neighbors $\partial_G A = N_{G/\sim}(v_A)$, then the broadcast for $v_A$
to a neighbor $w$ is the minimum over all such broadcasts from vertices in $A$.
Whenever the forwarding policy in $G$ had the form: ``receive from some $w$
with target $t$ at $x \in A$, forward through $A$ to some final node $y \in A$,
who forwards to $w'$,'' the forwarding policy in $G/\sim$ is: ``Forward
messages from $w$ with target $t$ to $w'$.''

Moreover, the concepts of forwarding number and colluding paths defined in
Section~\ref{sec:separated} for separated agents in the uniform model have
analogous definitions in the nonuniform model.  So when we say that a component
$A \subset V(G)$ has a minimal $j$-th colluding path, the $j$ refers to the
path in the quotient graph, which lifts to a path in $G$ (one of many, and
possibly involving many more than $j$ colluders).  The strategy of forwarding
along a minimal colluding path lifts from the quotient graph to a strategy that
forwards along paths between connected components.

With this understanding, the main strategy can be sketched as follows. Each
connected component of colluding agents $A \subset C$ determines a minimal
$j$-th colluding path $p_t$ for each target $t \in V(G)$ using the algorithm
from Section~\ref{sec:separated}. Pick any $x \in A$ which is adjacent to the
first honest vertex $w$ on $p_t$, call this the \emph{$t$-exit node} for $A$,
and have $x$ broadcast $\rho(x,t) = \rho(w,t) - 1$ as usual. If every
non-$t$-exit node in $A$ broadcasts so that the message never returns to $A$,
then the proof of Proposition~\ref{prop:rhostar-admissible} generalizes to
prove no cycles occur for this strategy. 

We now describe bounds on the minimality of such broadcasts. For $A \subset
V(G)$, denote by $d_{G - A}(x,y)$ the distance from $x$ to $y$ in the subgraph
induced by $V(G) - A$. When $A = \{ a \}$ is a single node we abuse notation
and write $d_{G-a}$. We further write $\rho_{G-A}$ to denote the
perceived/broadcast distances when $A$ is removed. Note these values change for
honest agents when paths are eliminated, but not for colluding agents. 

As a simple illustrative first case, suppose there are exactly two colluding
agents $x,y$ and they are adjacent. Let $t \in V(G) - \{ x,y \}$. If $y$
forwards a message to $x$, who in turn forwards to $t$ through $w \neq y$, then
in order to prevent the message cycling back though $y$, we require
$d_{G-x}(w,y) + \rho(y,t) \geq d(w,t)$, which rearranges to give a condition on
$y$'s broadcast. If $d_{G-x}(w,y) = \infty$ this is interpreted as no
restriction, and $\rho(y,t)$ may be 1.

For a connected component $A$ and target $t$, a similar bound is imposed on
every node in $A$ which is not the $t$-exit node. We state it as a theorem.

\begin{theorem}\label{thm:generalization}

Let $G$ be a graph, $C \subset V(G)$ be a subset of colluding agents whose
induced subgraph has components $C_1 \cup \dots \cup C_s$. For each component
$C_i$ and target $t$, pick a $t$-exit node $v_{i,t} \in C_i$, who behaves as
described above, and have every $x \in C_i$ forward messages with target $t$ to
$v_{i,t}$. Call $w_{i,t}$ the node that $v_{i,t}$ forwards to. Pick any
broadcast of the non-exit nodes $x \in C_i$ such that for all $j$ with $C_j$
having no larger forwarding number,

\[ 
    \rho(x,t) \geq \rho(w_{j,t},t) - d_{G-(C_j - \{ x \} )}(w_{j,t}, x),
\]

setting $\rho(x,t) = 1$ if all of the above bounds are nonpositive or
$-\infty$. 

This strategy is admissible.

\end{theorem}

\begin{proof}

As discussed above, it suffices to show that a message for $t$ forwarded
through $C_i$ to the $t$-exit node $v_{i,t}$ never returns to $A$. Suppose to
the contrary the message follows some path $p$ hitting $x \in C_i$. Without
loss of generality, $x \neq v_{i,t}$ is the first to receive the message. Now
$\rho(w_{i,t}, x) \geq d_{G-(C_i - \{ x \})}(w_{i,t}, x)$, and so by assumption
(that $x$ gets the message), they are equal and $w_{i,t}$ is in a tie-break
situation.
\end{proof}

This strategy is not optimal. Figure~\ref{fig:forwarding-ctex} gives a
counterexample, in which the central issue is that two components which are
tied for minimal forwarding number could improve their joint strategy by having
one component forward through the other. In contrast to the separated case, a
different choice of forwarding policy implies different broadcasts for the
nodes. Indeed, if $X$ is the space of all strategies induced by all possible
forwarding configurations and the implied broadcasts from
Theorem~\ref{thm:generalization}, it is easy to see that an optimal strategy is
a member of $X$.  Still, it is unclear whether it is NP-hard to pick the
optimal forwarding policy.

\begin{figure}
\includegraphics[width=0.5\textwidth]{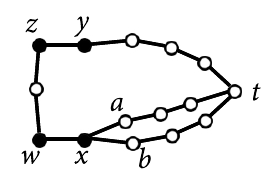}
\caption{A counterexample to the optimality of our strategy in the setting
where agents can be adjacent. Shaded nodes are colluding. If the component with
$z,y$ is processed first then our algorithm correctly chooses the $w,x$
component to have forwarding number 2 (with $x$ broadcasting 1), capturing all
traffic from $a,b$ to $t$.  On the other hand if $w,x$ is processed first the
result will miss traffic from $a,b$.} \label{fig:forwarding-ctex}
\end{figure}

\section{Simulations} \label{sec:simulations}
 
We simulated the protocols described in this paper on four networks (Figure
\ref{fig:erdos-renyi}). (The code used to run the experiments is available at
https://github.com/j2kun/information-monitoring-simulations.) The first is an
Erd\"os-R\'enyi random graph $G(n, p)$ where $n = 1000, p = 4/1000$. The second
is a preferential attachment model~\cite{BarabasiA99} with 100 nodes. The third
is a Watts-Strogatz model~\cite{WattsS98} with $n=1000$, degree $k$=10, and
edge-rewiring probability $\beta=0.04$. The fourth is a snapshot of the
Autonomous Systems (AS) subgraph of the United States, which has 9,296 nodes
and 17,544 edges. The AS graph comes from the website of
Newman~\cite{Newman06}. 

For each network, we inspect the potential advantage of the $\rho^*$ strategy
of Section~\ref{sec:separated} over the strategy in which all agents act
independently according to $\rho(x,t) = d(x,t) - 2$, and for each strategy we
also compared the performance for a randomly chosen subset of nodes versus
nodes of highest degrees. For various number of colluding agents, we report the
number of paths intercepted by the colluding agents in the worst-case scenario,
i.e., when there are multiple  shortest paths (some of which may be ``perceived
shortest paths'') between the source and target of a message, the message would
go through the paths with no colluding agents.

It is worth noting that for each of these cases, we assume the demand between
every pair of nodes is uniform. That is, there is no advantage for the
colluders to intercept one shortest path over another.

Often, large degree nodes tend to be better colluders than randomly chosen
nodes. Regardless, the benefit of colluding is clear even for randomly chosen
nodes. In fact for the US AS network, if we compromise only 18 random nodes
(roughly 0.2\%) we can intercept an expected 10\% of the entire network's
traffic. For smaller percentages, the randomly selected colluding nodes capture
more traffic in the US AS network than in the synthetic models. However, for
larger percentages the amount of captured traffic in the US AS network levels
off quite dramatically, revealing additional topological structure in the US AS
network that is not present in the synthetic models. It is also interesting to
note the relatively small difference between the two strategies in the US AS
graph, implying that in this setting collusion does not provide a
quantitatively large improvement over a simpler strategy. We stress that since
the $\rho^*$ strategy (as defined in Section~\ref{sec:separated}) may not be
optimal when there are adjacent colluding nodes, the plots are lower bounds for
the amount of traffic captured by the optimal strategy with a given percentage
of colluding nodes. For uniform broadcasts and small random colluding sets, our
estimates are fairly accurate because with high probability none of the
colluders are adjacent. On the other hand, for nonuniform broadcasts the
optimal strategy for 18 random nodes could very well capture more than 10\% of
messages on the AS networking with uniform traffic.

\begin{figure*}[htbp]
\centering
\subfigure{
\includegraphics[width=0.45\textwidth]{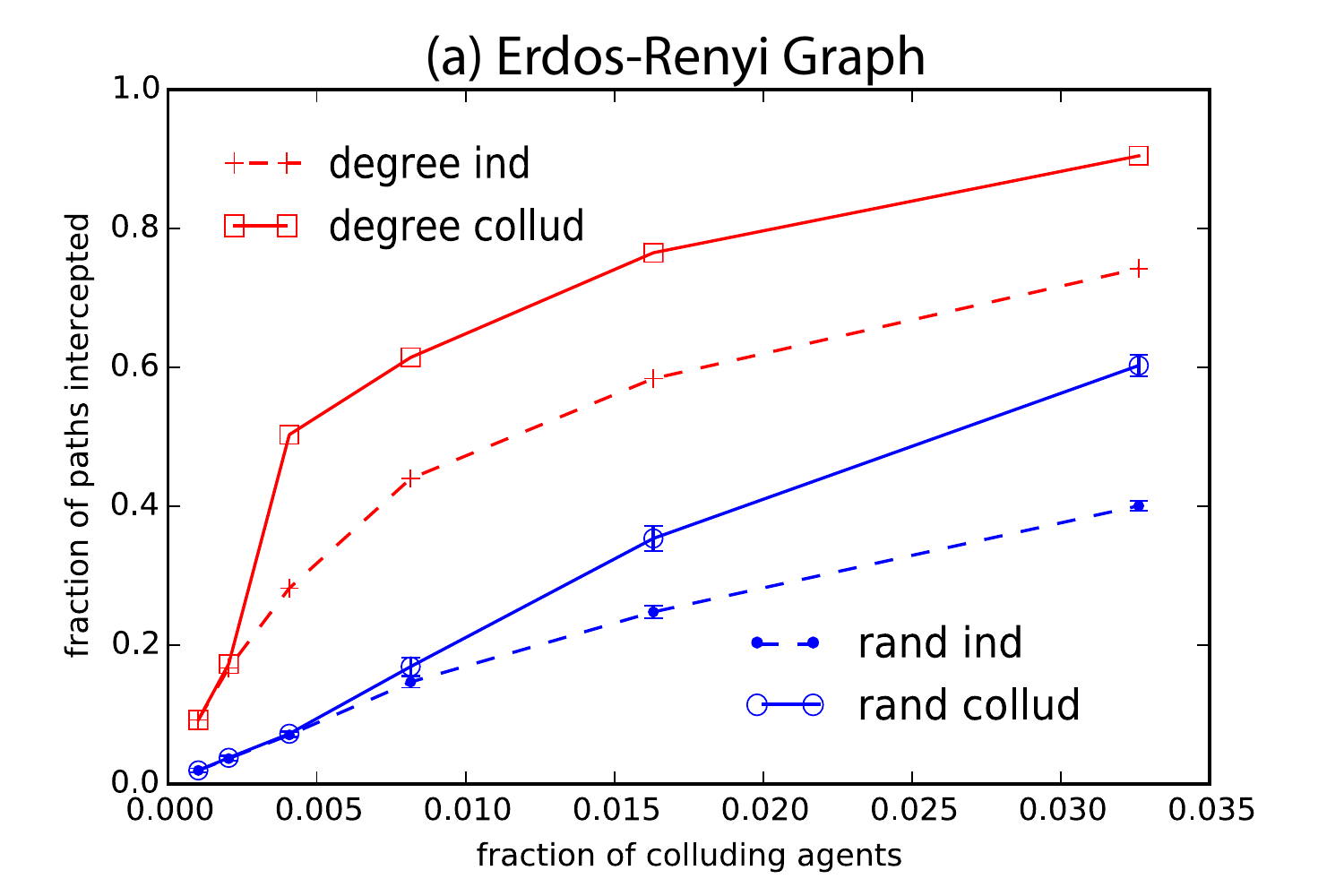}
}
\quad
\subfigure{
\includegraphics[width=0.45\textwidth]{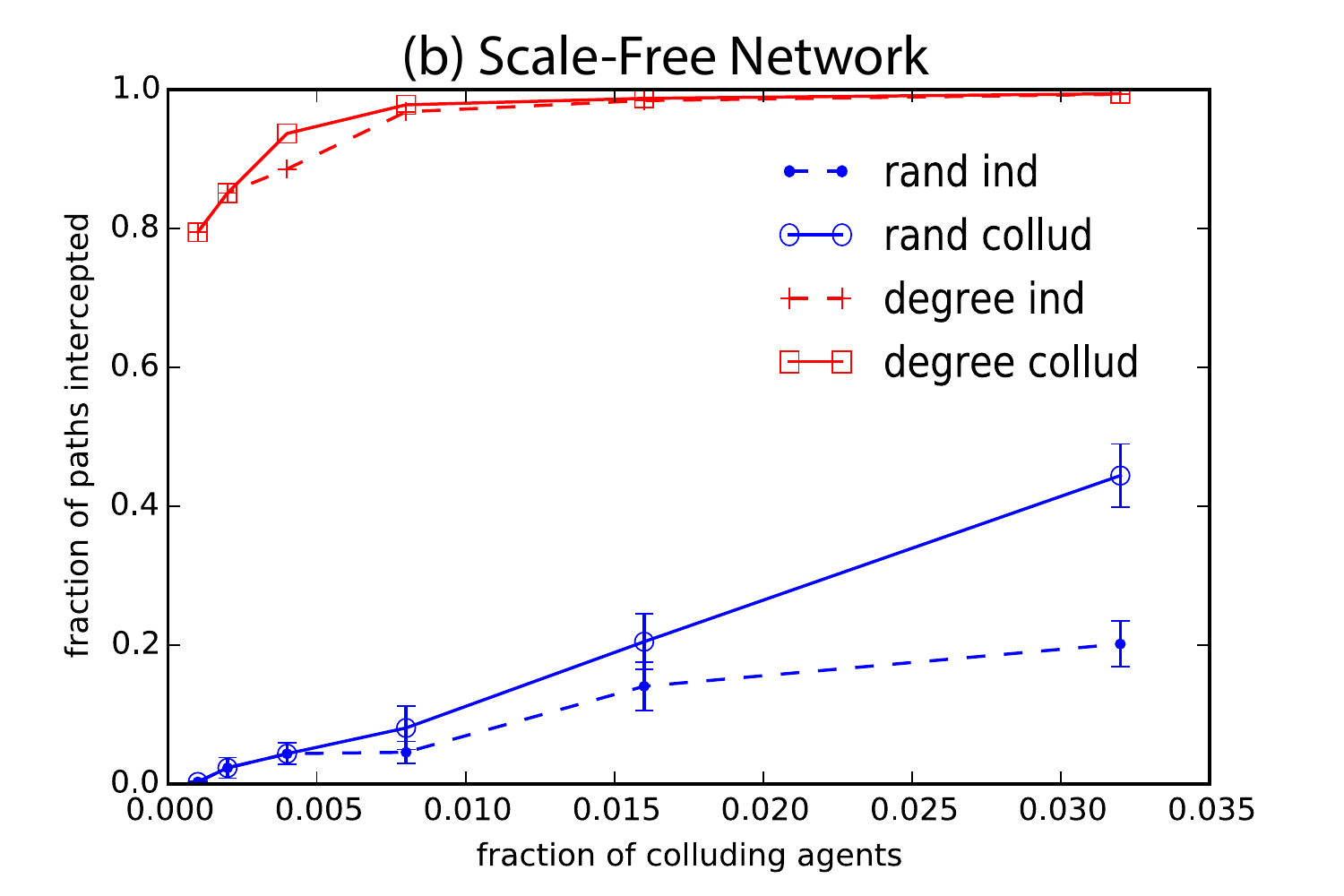}
}
\\
\subfigure{
\includegraphics[width=0.45\textwidth]{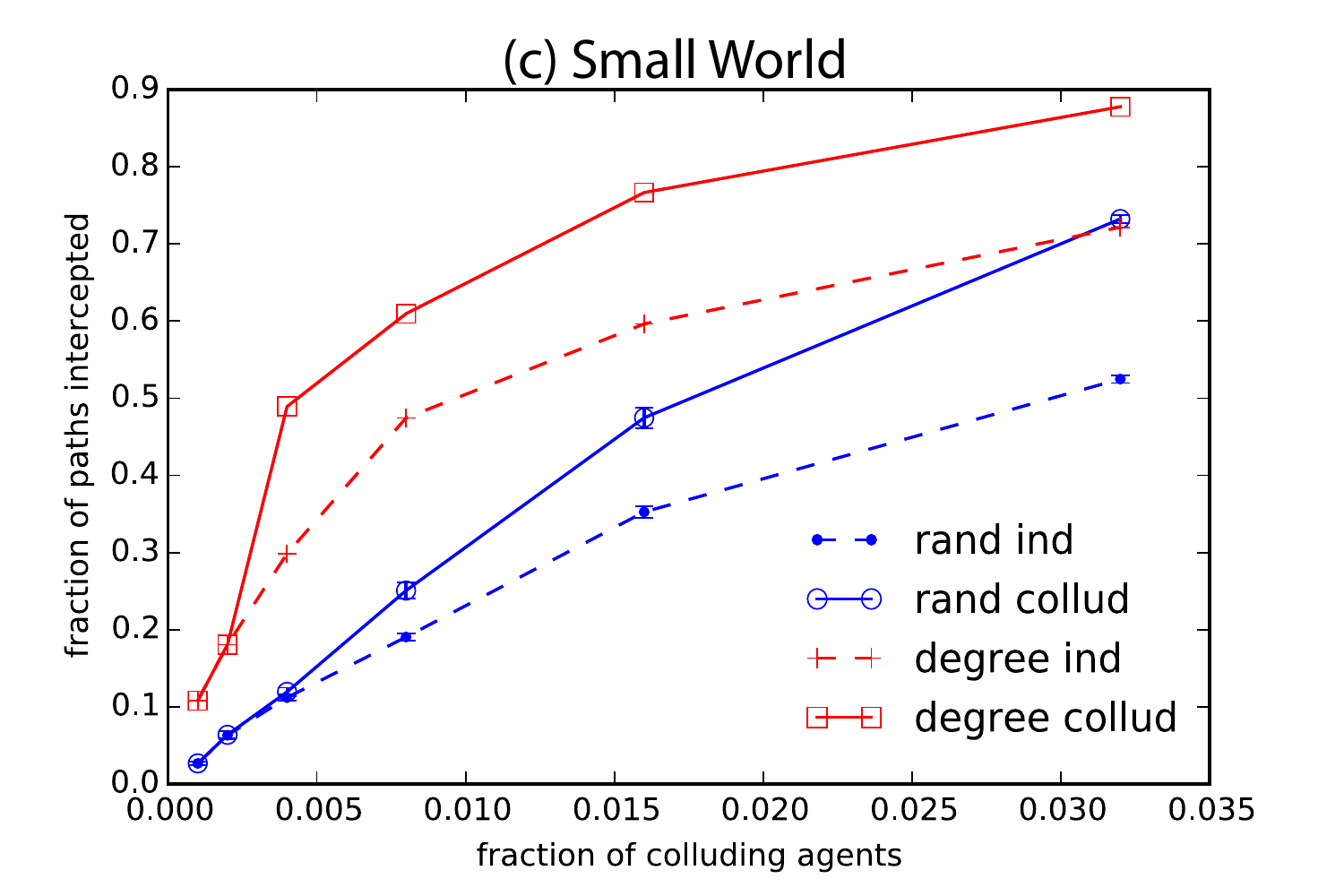}
}
\quad
\subfigure{
\includegraphics[width=0.45\textwidth]{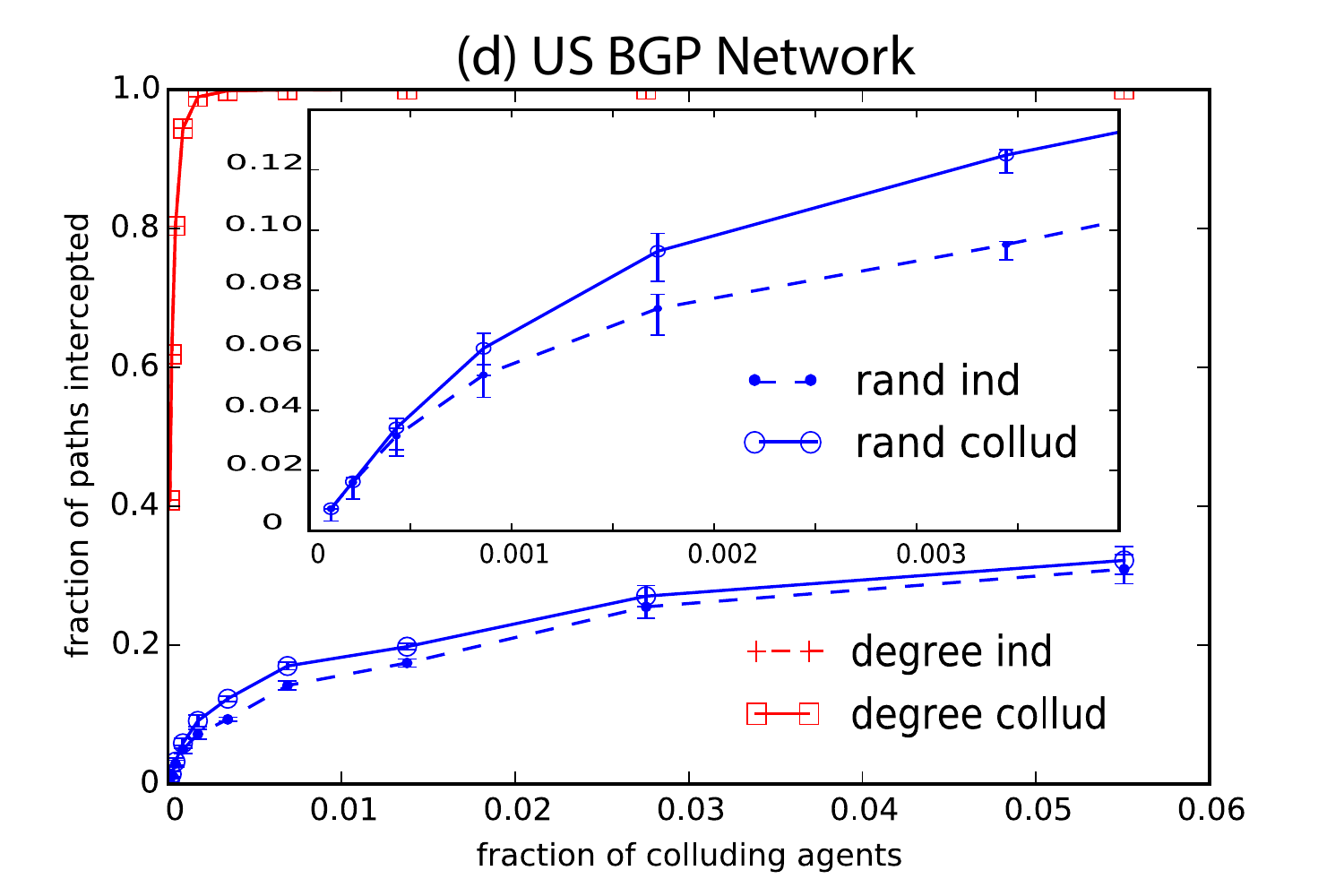}
}


\caption{Fraction of paths intercepted by a varying number of colluding agents
on a Erd\"os-R\'enyi random graph (top left), a preferential attachment graph
(top right), a Watts-Strogatz model (bottom left), and the US AS network
(bottom right). Blue curves correspond to subsets of colluding agents chosen
uniformly at random, while red curves to subsets chosen by largest degrees.
Dashed curves correspond to the strategy where each agent independently lies by
an additive factor of two, while solid curves to the optimal separated strategy
of Section~\ref{sec:separated}. The inset graph for the AS network magnifies
the leftmost portion of the blue curves.}
\label{fig:erdos-renyi}
\end{figure*}

\section{Discussion and open problems} \label{sec:conclusion}

In this paper we introduced and related two variants of a model of traffic interception in routing networks with distance-vector protocols, one for uniform broadcasts and one for non-uniform
broadcasts. We characterized the optimal strategy for the uniform setting in
which no two colluding agents are adjacent, and provided a family of strategies
for the general case. We simulated the impact of the optimal strategy in the
uniform broadcasting model on an assortment of graphs and found that in
expectation for the US Autonomous Systems network, randomly selecting 0.2\% of
the nodes to act as colluding agents captured 10\% of the entire network
traffic paths.

There are a number of directions for future work.  For example, we could consider alternate
definitions of admissibility.  In this work we defined strategies to be admissible 
as long as there exists a corresponding path between every pair of nodes, motivated by the model assumption
that honest nodes break ties uniformly at random.  Instead we can consider different definitions of admissibility, 
where for every pair of nodes {\em{all} }corresponding paths that start and end with those nodes must also be simple paths.   
In the case of the uniform model with a single colluding agent $x$,broadcasting the distance to $t$, $\rho(x,t)=\min(d(x,t)-1,1)$ 
satisfies this more restrictive definition of admissibility.   Alternatively, we may call a strategy admissible if
the length of a typical corresponding path under the strategy does not drastically differ from the length of a typical corresponding path under the honest strategy.  The last admissibility definition is not meant to be taken as a rigorous one; we only wish to illustrate the diversity of these alternate admissibility definitions.

We could also ask  whether one can efficiently characterize the general case of adjacent agents,
or whether deciding the appropriate forwarding mechanism is NP-hard. In either
case, another open direction is to provide approximation algorithms when the
optimal subset of colluding agents is unknown. While this is likely to correlate with
betweenness centrality we would also be interested in finding the
subset of agents with the largest \emph{relative} improvement. A further
improvement to this work may also consider realistic network traffic, as
opposed to considering an absolute number of traffic paths as the interception
objective.

\section*{Acknowledgement} 

We thank everyone involved in the 2014 AMS Network Science Mathematical
Research Community for inspiration and many helpful comments, especially the
organizers Aaron Clauset, David Kempe, and Mason Porter. We also thank Lev
Reyzin for his helpful comments.

Research supported in part by NSF grants 
CMMI-1300477 and CMMI-1404864. 

\bibliographystyle{plain}
\bibliography{main}

\end{document}